\documentclass[submission,copyright,creativecommons]{eptcs}
\providecommand{\event}{AFL 2023} 

\usepackage{iftex}
\usepackage{amsmath}
\usepackage{amsthm}
\usepackage{amssymb}
\usepackage{graphicx}

\newtheorem{theorem}{Theorem}
\newtheorem{lemma}{Lemma}

\theoremstyle{definition}
\newtheorem{example}{Example}
\newtheorem{definition}{Definition}

\ifpdf
  \usepackage{underscore}         
  \usepackage[T1]{fontenc}        
\else
  \usepackage{breakurl}           
\fi

\title{Separating Words from Every Start State\\ with Horner Automata}
\author{Nicholas Tran
\institute{Santa Clara University}
\institute{Santa Clara, CA 95053}
\email{ntran@scu.edu}
}
\def\titlerunning{Separating words from every start state with Horner automata}
\def\authorrunning{N. Tran}
\begin{document}
\maketitle

\begin{abstract}
We show that a well-known family of deterministic finite automata $H_{b, m}$ can be used to distinguish distinct binary strings of the same length from every start state.  Further, we establish a lower bound of $\Omega(\sqrt{n/\log n})$ and an upper bound of $O(\sqrt{n}\log n\log\log n)$ on the number of states of $H_{b, m}$ necessary to achieve this type of separation. Our latter result improves the currently best known $O(n)$ upper bound for arbitrary DFA.
\end{abstract}

\section{Introduction}
Given two distinct strings (or words) over some alphabet, we are interested in the minimum size of deterministic finite automata that end in different states after reading the strings {\em for every (common) start state}.  We call this the {\em $\forall$-separation distance} between strings, and  when the alphabet is nonunary, we use $D_{\forall}(n)$ to denote the largest $\forall$-separation distance between two distinct strings of length $n$.  This variant of the separating words problem was recently introduced and studied in \cite{tranCIAA22}, where a lower bound of $\Omega(\log n)$ and an upper bound of $n+1$ on $D_{\forall}(n)$ were established.  The main result of this paper is an improved $O(\sqrt{n}\log n\log\log n)$ upper bound on $D_{\forall}(n)$.

The original separating words problem was studied in \cite{GK,R,Demaine,Vyalyi2014,ChaseSTOC}.  The standard notion of separation words by deterministic finite automata requires one string to be accepted and the other rejected.  It is clearly equivalent to the definition of separation stated  initially in \cite{GK} 
that does not involve accepting states: a deterministic finite automaton separates two strings if it ends in different states after reading the strings {\em for some (common) start state}.  
We use $D_{\exists}(n)$ to denote the largest separation distance of this type between two distinct strings of length $n$ for nonunary alphabets.  The best known upper bound on $D_{\exists}(n)$ is $O(n^{1/3}\log^7 n)$ \cite{ChaseSTOC}, and the best known lower bound on $D_{\exists}(n)$ is $\Omega(\log n)$ \cite{Demaine}.  It is known that the separation distances are tightly bound by $\log n$ for strings of different lengths $m < n$  (in particular, for distinct strings over unary alphabets), and that $D_{\exists}(n)$ and $D_{\forall}(n)$ do not depend on the (nonunary) alphabet size. Table \ref{tab:values} lists values of $D_{\exists}(n)$ and $D_{\forall}(n)$ for small values of $n$, obtained via exhaustive search\footnote{These results were computed using a C++ program running on a Linux workstation with Intel® Core™ i7-4790S CPU @ 3.20GHz and 16 GB RAM.  
}.

\begin{table}[h]
  \centering
  \begin{tabular}{c|c|c|c|c|c|c|c|c|c|c|c|c|c|c|c|c|c|c}
    \hline
    $n$ & 1 & 2 & 3 & 4 & 5 & 6 & 7 & 8 &9 &10&11&12&13&14&15&16&17&18\\
    \hline
    $D_{\exists}(n)$ & 2 & 2 & 2 & 3 & 3 & 3 & 3 & 3 & 3 & 4 & 4 & 4 &4  & 4 &4 &4 & 4 & 5\\
    \hline
    $D_{\forall}(n)$ & 2 & 2 & 3& 3 & 3 & 3 & 3 & 4 & 4 & 4 & 4 & 4 &4  & 4 &5 &5 & 5 & 5\\
    \hline
  \end{tabular}
  \caption{Values of $D_{\exists}(n)$ and $D_{\forall}(n)$ for $1\le n \le 18$.}
  \label{tab:values}
\end{table}

If the separating automaton is required to end in different states for {\em every pair} of start states, then the largest so-called $\forall^2$-separation distance between two distinct strings of length $n$ is exactly $n+1$ and unbounded between strings of different lengths (i.e., they may not be separable).  At the other extreme, if the separating automaton is required to end in different states for {\em some pair} of start states, then the largest so-called $\exists^2$-separation distance between two distinct strings (regardless of lengths) is 2.  These and other results can be found in \cite{tranCIAA22}.  

The best lower bound of $\Omega(\log n)$ on $D_{\exists}(n)$ applies trivially to $D_{\forall}(n)$; it is not known if a stronger lower bound holds for the latter.  On the other hand, the best upper bound of $O(n^{1/3}\log^7 n)$ on $D_{\exists}(n)$ does not readily apply to $D_{\forall}(n)$.  In fact, it is not immediately clear how to establish even a looser upper bound such as $n+1$.  In this paper we show how to adapt the technique of 
counting the number of cyclotomic factors of certain polynomials 
\cite{Vyalyi2014,ChaseSTOC} to obtain an $O(\sqrt{n}\log n\log\log n)$ upper bound on $D_{\forall}(n)$. Our result improves the currently best known $O(n)$ upper bound obtained in \cite{tranCIAA22}. 

Specifically, we interpret each nonempty binary string $s$ as the representation of an integer $N_{s, b}$ in some base $b$ and show how to compute $N_{s, b} \bmod{m}$ for some modulus $m$ using a familiar deterministic finite automaton $H_{b, m}$ with $m$ states and dependent only on $b$ and $m$.  
We then show that if  $s$ and $t$ are distinct binary strings of length $n \ge 1$, then $N_{b, s} \not\equiv N_{b, t} \pmod{m}$ for some $0 \le b < m \in O(\sqrt{n}\log n\log\log n)$.  Hence $H_{b, m}$ $\exists$-separates $s$ and $t$, and in fact we show that $H_{b, m}$ $\forall$-separates $s$ and $t$.
On the other hand, we show that for every $n \ge 1$, there are distinct binary strings $s$ and $t$ of length $n$ such that the smallest $H_{b, m}$ that $\forall$-separates $s$ and $t$ have $\Omega(\sqrt{n/\log n})$ states.

The rest of this paper is organized as follows.  Section 2 reviews basic definitions about separating words with automata and states some useful number-theoretic facts.  The next section contains the main results: Subsection 3.1 presents so-called Horner automata $H_{b, m}$ and shows how to use them to $\forall$-separate strings, Subsection 3.2 proves an $\Omega(\sqrt{n/\log n})$ lower bound on the size of $H_{b, m}$ required to $\forall$-separate two distinct binary strings of length $n$, and Subsection 3.3 establishes an $O(\sqrt{n}\log n\log\log n)$ corresponding upper bound. Section 4 discusses ideas for future work.

\section{Preliminaries}

The symbols of a string $s$ of length $n \ge 1$ from left to right are denoted by $s_0, s_1, \ldots, s_{n-1}$. The natural and binary logarithms are denoted by $\ln$ and $\log$ respectively. 
We use the following simplified definition of deterministic finite automata that does not specify an initial state or accepting states:

\begin{definition}
  A {\em deterministic finite automaton} (DFA) is a triple $$M = (Q, \Sigma, \delta),$$ where $Q$ is a finite set of {\em states}, $\Sigma$ is an alphabet,  and $\delta : Q\times \Sigma \rightarrow Q$ is a {\em transition} function.  We use $|M|$ to denote the number of states of $M$ and refer to it as the {\em size} of $M$.

  The {\em extended transition function} $\delta': Q\times \Sigma^* \rightarrow Q$ is defined recursively:
  \begin{enumerate}
    \item $\delta'(q, \epsilon) = q$, where $\epsilon$ is the empty string, for $q\in Q$;
    \item $\delta'(q, xa) = \delta(\delta'(q, x), a)$ for $a\in \Sigma$, $x\in \Sigma^*$ and $q\in Q$.
  \end{enumerate}
The first and last states in the sequence of states that $M$ enters when reading a string from left to right are called the {\em start} and {\em end} state respectively.
\end{definition}
\begin{definition}
  \label{def:separation}
  Let $x$ and $y$ be strings over an alphabet $\Sigma$.  We say a DFA $M$
  \begin{itemize}
  \item {\em $\exists$-separates} 
  $x$ and $y$ if $\delta'(s, x) \not= \delta'(s, y)$ for {\em some}  $s\in Q$; 
  
  \item {\em $\forall$-separates} $x$ and $y$ if $\delta'(s, x) \not= \delta'(s, y)$ for {\em every}  $s\in Q$.

  \end{itemize}
  
  \end{definition}

\begin{example} Let $s = 0000\ 0000$ and $t = 1111\ 1100$.  

  The two-state DFA in Fig.~\ref{fig:ex-sep} (left) $\exists$-separates $s$ and $t$ because $a = \delta'(a, s) \not= \delta'(a, t) = b$, and it is clearly a smallest such automaton.  On the other hand, this DFA does {\em not} $\forall$-separate $s$ and $t$ because $\delta'(b, s) = \delta'(b, t) = b$.
  
  The four-state DFA in Fig.~\ref{fig:ex-sep} (right) $\forall$-separates $s$ and $t$ because
  \begin{eqnarray*}
  a = \delta'(a, s)&\not=&\delta'(a, t) = b,\\
  b = \delta'(b, s)&\not=&\delta'(b, t) = a,\\
  c = \delta'(c, s)&\not=&\delta'(c, t) = a,\\
  d = \delta'(d, s)&\not=&\delta'(d, t) = a.
  \end{eqnarray*}
It is shown to be a smallest such automaton in \cite{tranCIAA22}.
\begin{figure}[h]
  \hskip 0.5in\includegraphics[height=.75in]{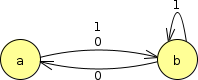}
\hfill \includegraphics[height=1.5in]{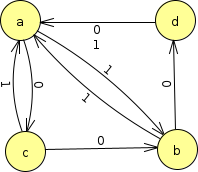}\hskip 0.5in\hfill
\label{fig:ex-sep}
\caption{The DFA on the left $\exists$-separates $0000\ 0000$ and $1111\ 1100$ but does not $\forall$-separate them.  The DFA on the right $\forall$-separates the strings.}
\end{figure}
\end{example}
The following facts can be found in standard texts on number theory, e.g., \cite{niven}.  The {\bf Prime Number Theorem} states that $\pi(x)$, the number of primes less than or equal to $x$, is approximately $x/\ln x$, i.e., $$\lim_{x\to \infty} \frac{\pi(x)}{\frac{x}{\ln x}} = 1.$$ 
{\bf Bertrand's postulate} states that for every $n > 1$, there is a prime $p$ such that $n < p < 2n$.  
For every integer $n \ge 1$, the $n^{\rm th}$ {\bf cyclotomic polynomial} $\Phi_n(x)$ is defined as the polynomial whose zeros are the {\em primitive} $n^{\rm th}$ roots of unity:
$$\Phi_n(x) = \prod_{\substack{1 \le k \le n \\ \gcd(k, n) = 1}} (x - e^{2\pi i k/n}).$$
We list relevant properties of cyclotomic polynomials here.  The coefficients of $\Phi_n(x)$ are integers, and its degree is $\phi(n)$, 
the {\bf Euler's totient function}; it is known that
$\phi(n) \in \Omega(n/\log\log n)$ \cite{RosserSchoenfeld} and that the function $x/\log\log x$ is increasing when $x \ge 6$ (e.g., Wolfram|Alpha).  The cyclotomic polynomials are irreducible and co-prime over $\mathbb{Q}$, i.e., $\gcd(\Phi_n(x), \Phi_m(x)) = 1$ for $n\not= m$.  Finally, $\Phi_n(x)$ divides $x^n-1$ for all $n \ge 1$.

\begin{example}
The properties listed above can be seen to hold for the first few cyclotomic polynomials:
\begin{eqnarray*}
\Phi_1(x) &=& x-1,\\
\Phi_2(x) &=& x+1,\\
\Phi_3(x) &=& x^2+x+1,\\
\Phi_4(x) &=& x^2+1,\\
\Phi_5(x) &=& x^4+x^3+x^2+x+1,\\
\Phi_6(x) &=& x^2-x+1.
\end{eqnarray*}
\end{example}

The set of integers modulo $m$ is denoted $\mathbb{Z}_m$, and the set of polynomials in $x$ with integer coefficients is denoted $\mathbb{Z}[X]$.  A polynomial $P(x) \in \mathbb{Z}[X]$ is said to {\bf vanish} modulo $m$ if $P(x) \equiv 0 \pmod{m}$ for all $x \in \mathbb{Z}$.   
{\bf Lagrange's theorem} states that if $P(x) \in \mathbb{Z}[X]$ has degree $n \ge 1$ and $p$ is a prime,  
then either $P(x)$ has at most $n$ zeros modulo $p$, or all coefficients of $P(x)$ are divisible by $p$.  By {\bf Fermat's little theorem}, $x^p - x$ vanishes modulo $p$ when $p$ is prime. 

We associate with each {\em binary} string $s$ of length $n \ge 1$ the polynomial $$s(x) = \sum_{j=0}^{n-1} s_j x^{n-1-j}$$ with coefficients in $\mathbb{Z}_2$.  If $s$ and $t$ are distinct binary strings of  length $n \ge 1$, $s(x) -t(x)$ is a nonzero polynomial of degree at most $n-1$ with coefficients in $\{-1, 0, 1\}$.

\begin{example}
  Again, let $s = 0000\ 0000$ and $t = 1111\ 1100$.  The associated polynomials are $s(x) = 0$ and $t(x) = x^7+x^6+x^5+x^4+x^3+x^2$.  Their difference $s(x) - t(x)$ is $-(x^7+x^6+x^5+x^4+x^3+x^2)$.

  This difference polynomial vanishes modulo $2$ because 
\begin{eqnarray*}
  -(0^7+0^6+0^5+0^4+0^3+0^2) = 0 &\equiv& 0 \pmod{2},\\
  -(1^7+1^6+1^5+1^4+1^3+1^2) = -6 &\equiv& 0 \pmod{2}.
\end{eqnarray*}
  Similarly, it also vanishes modulo $3$ because
\begin{eqnarray*}
  -(0^7+0^6+0^5+0^4+0^3+0^2) = 0 &\equiv& 0 \pmod{3},\\
  -(1^7+1^6+1^5+1^4+1^3+1^2) = -6 &\equiv& 0 \pmod{3},\\
  -(2^7+2^6+2^5+2^4+2^3+2^2) = -252 &\equiv& 0 \pmod{3}. 
\end{eqnarray*}
However, $s(x) - t(x)$ does not vanish modulo $5$ because
$$-(1^7+1^6+1^5+1^4+1^3+1^2) = -6 \not\equiv 0 \pmod{5}.$$
We will see in the next section that the above observation can be deduced by noting that $$x^5-x = x(x-1)(x+1)(x^2+1)$$ does not divide $s(x) - t(x)$ over the rationals:
\begin{eqnarray*}
  -(x^7+x^6+x^5+x^4+x^3+x^2)  &=&-x^2(x+1)(x^2-x-1)(x^2+x+1).
\end{eqnarray*}
\label{example:3}
\end{example}

\section{Main Results}
In this section we introduce Horner automata and show how to use them to $\forall$-separate strings.  We then establish almost matching lower bound and upper bound on the size of the smallest Horner automata that $\forall$-separate two distinct binary strings of length $n$.
\subsection{Horner automata and $\forall$-separation}
For $0 \le b < m$, let $H_{b, m}$ be the deterministic finite automaton with $m$ states $0$, $1$, $\ldots$, $m-1$ and transition function $\delta(i, a) = (ib + a)\bmod {m}$ for $a \in \{0, 1\}$; note that the input alphabet is binary.  The special cases $H_{2, m}$ are usually introduced in an introductory course on automata theory to recognize binary strings representing integers divisible by $m$.  They are named {\bf Horner automata} in \cite{Sutner}, because on binary input $s$ and start state $0$, they compute the value of the associated polynomial $s(b)$ using Horner's rule and end in state $s(b) \bmod{m}$:
$$\delta'(0, s) = ((\cdots((0b + s_0)b + s_1)b +\cdots  + s_{n-2})b + s_{n-1}) \bmod{m} = s(b) \bmod{m}.$$
In other words, these automata ``compute'' $N_{s, b} \bmod m$, where $N_{s, b}$ is the integer represented by {\em binary} string $s$ in base $b$ (but they are not the smallest ones to do so \cite{Alexeev}).  We prove a slightly stronger statement of this fact in the following lemma.

\begin{lemma}
Let $0 \le b, i < m$ be integers and $s$ be a binary string of length $n \ge 1$.  Starting in state $i$, the Horner automaton $H_{b, m}$ ends in state $(ib^{n} + s(b)) \bmod{m}$ after reading $s$.
\label{lemma:1}
\end{lemma}
\begin{proof}
By induction on $n$.  When $n = 1$, the string $s$ is just symbol $s_0$, and the associated polynomial $s(x)$ is the constant polynomial $s_0$.  Starting in state $i$, $M_{b, m}$  ends in state $$\delta'(i, s_0) = (ib + s_0) \bmod{m} = (ib^1 + s(b)) \bmod{m},$$
after reading $s$, so the lemma holds.  

Assume that the lemma holds for $n$ and let $s = s_0s_1\ldots s_n$ be a binary string of length $n+1$.  Starting in state $i$ on input $s$, the automaton $H_{b, m}$ ends in state
\begin{eqnarray*}
\delta'(i, s)&=&\delta'(i, s_0\ldots s_{n-1}s_n)\\
&=&\delta(\delta'(i, s_0\ldots s_{n-1}), s_n)\\
&=&\delta((ib^n + \sum_{j=0}^{n-1}s_jb^{n-1-j}) \bmod{m}, s_n)\\
&=&(ib^{n+1} + \sum_{j=0}^{n-1}(s_jb^{n-j}) + s_n) \bmod{m}\\
&=&(ib^{n+1} + \sum_{j=0}^{n}s_jb^{n-j}) \bmod{m}\\
&=&(ib^{n+1} + s(b)) \bmod{m}.
\end{eqnarray*}
\end{proof}
\begin{example}
Fig. \ref{fig:h25} shows the automaton $H_{2, 5}$ on the left.  On input $s = 10\ 1111$ and start state $0$, it ends in state $s(2) \bmod{5} = 2^5+2^3+2^2+2^1+2^0 = 47 \bmod{5} = 2$.  In contrast, on input $t = 11\ 1011$ and start state $0$, it ends in state $t(2) \bmod{5} = 2^5+2^4+2^3+2^1+2^0 = 59 \bmod{5} = 4$.  Thus, $H_{2, 5}$ $\exists$-separates $s$ and $t$, and in fact, it $\forall$-separates $s$ and $t$ due to Lemma \ref{lemma:1}.  We demonstrate in general this important property of Horner automata below.  

\begin{figure}[h]
  \centerline{\hskip 1in\includegraphics[height=1.5in]{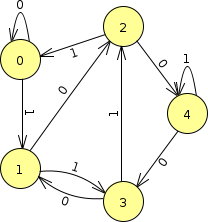}\hfill \includegraphics[height=0.65in]{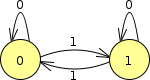}\hskip 1in}
  \caption{Horner automata  $H_{2, 5}$ (left) and $H_{1, 2}$ (right)}\label{fig:h25}
\end{figure}
\end{example}

It is natural to study the separation distance between two binary strings of the same length by Horner automata.  We define this notion formally and study its properties here.
\begin{definition}\ 
\begin{enumerate}
  \item The Horner distance between two distinct binary strings $s$ and $t$, denoted by $d_H(s, t)$, 
  is the smallest $m$ such that $H_{b, m}$ $\exists$-separates $s$ and $t$ for some $0 \le b < m$.
  \item The Horner separation distance, denoted by $D_H(n)$, is the maximum Horner distance over all pairs of distinct binary strings $s$ and $t$ of length $n$, for $n \ge 1$.
\end{enumerate}
\end{definition}



We now show that the Horner distance between two distinct binary strings $s$ and $t$ of the same length is always defined, and furthermore, it is an upper bound of the $\forall$-distance between $s$ and $t$.


\begin{lemma} Let $s$ and $t$ be binary strings of length $n$.  The following are equivalent:
  \begin{enumerate}
    \item $s$ and $t$ are distinct;
    \item $s$ and $t$ are $\exists$-separated by $H_{b, m}$ for some $0 \le b < m \le 2n $;
    \item $s$ and $t$ are $\forall$-separated by $H_{b, m}$ for some $0 \le b < m \le 2n$.
  \end{enumerate}
  \label{lemma:2}
\end{lemma}

\begin{proof}
  (1) $\Rightarrow$ (2): Let $s$ and $t$ be two distinct binary strings of length $n$.  When $n = 1$, there is only one pair of distinct strings $0$ and $1$, and they are $\exists$-separated by $H_{1, 2}$ (shown in Fig.~\ref{fig:h25} on the right) when started in state 0.  
  For $n > 1$, by Bertrand's postulate, there is a prime $p$ such that $n < p < 2n$.  The polynomial $d(x) = s(x) - t(x)$ has degree at most $n-1$, coefficients in $\{-1, 0, 1\}$, and not all coefficients are zero, because $s$ and $t$ are distinct.  Thus, by Lagrange's theorem $d(x)$ has at most $n-1$ zeros in $\mathbb{Z}_p$, so there exists $0 \le b < p$ such that $d(b) \not\equiv 0 \pmod{p}$.  Since $d(b) = s(b) - t(b)$, we have $s(b) - t(b) \not\equiv 0 \pmod{p}$, or $s(b) \not\equiv t(b) \pmod{p}$.  Thus, $H_{b, p}$ $\exists$-separates $s$ and $t$ on start state $0$.
  
  (2) $\Rightarrow$ (3): Suppose $H_{b, m}$ separates binary strings $s$ and $t$ of length $n$ on start state $i_0$ for some $0 \le b, i_0 < m \le 2n$.  By Lemma \ref{lemma:1}, $i_0b^n + s(b) \not\equiv i_0b^n + t(b)\pmod{m}$, so $ib^n + s(b) \not\equiv ib^n + t(b) \pmod{m}$ for all $0 \le i < m$.  Thus, $s$ and $t$ are $\forall$-separated by $H_{b, m}$.

  (3) $\Rightarrow$ (1): Suppose $s$ and $t$ are binary strings of length $n$ that are $\forall$-separated by $H_{b, m}$  for some $0 \le b < m \le 2n$.  Then $s(b) \not\equiv t(b) \pmod{m}$, so $s(x) \not= t(x)$, and hence $s \not= t$.
\end{proof}

\subsection{Lower bound on the Horner separation distance $D_H(n)$}
We use a simple information-theoretic argument to show that the Horner separation distance $D_H(n)$ is at least $\Omega(\sqrt{n/\log n})$.

\begin{theorem}
$D_H(n) \in \Omega(\sqrt{n/\log n})$ for $n \ge 1$.
\end{theorem}
\begin{proof}
Let $s$ and $t$ be distinct binary strings of length $n$.  If they cannot be $\exists$-separated by a Horner automaton of size $M$ or less, the values of their polynomials $s(x)$ and $t(x)$ must be congruent modulo $m$ for all $0 \le b < m$ and $2 \le m \le M$.  The congruence classes of these values can be encoded in a binary string (called a {\em signature}) of length at most $$\sum_{m=2}^M\sum_{b=0}^{m-1}(\log m) \le \sum_{m=2}^M\sum_{b=0}^{m-1}(\log M) \in O(M^2\log M) \in O(M^2\log n)$$ due to the upper bound on $M$ in terms of $n$ given by Lemma \ref{lemma:2}.  There are $2^{O(M^2\log n)}$ such signatures and $2^n$ binary strings of length $n$, so to ensure that different $s$ and $t$ have different signatures and $\exists$-separable by a Horner automaton of size $D_H(n)$, we must have $2^n \le 2^{O(D_H^2(n)\log n)}$, or $n \in O(D_H^2(n)\log n)$, and hence $D_H(n) \in \Omega(\sqrt{n/\log n})$.
\end{proof}

\subsection{Upper bound on the Horner separation distance $D_H(n)$}

We begin with an observation that for large enough prime $p$, a polynomial $P(x)$ whose coefficients are in $\{-1, 0, 1\}$ vanishes modulo $p$ if and only it is divisible by $x^p-x$.

\begin{lemma}
Let $P(x) = \sum_{j=0}^n p_jx^j$ be a polynomial of degree $n \ge 0$ with coefficients $p_j \in \{-1, 0, 1\}$, and let $p \ge 1 + \sqrt{n}$ be a prime.  The polynomial $P(x)$ vanishes modulo $p$ if and only if $P(x) = (x^p-x)Q(x)$ for some $Q(x) \in \mathbb{Z}[x]$.
\label{lemma:3}
\end{lemma}
\begin{proof}
Suppose $P(x) = (x^p - x)Q(x)$ for some polynomial $Q(x) \in \mathbb{Z}[x]$.  By Fermat's little theorem, $x^p - x$ vanishes modulo $p$, so $P(x)$ also vanishes modulo $p$.

Conversely, suppose $P(x)$ vanishes modulo $p$.  Since $P(0) = p_0 \equiv 0 \pmod{p}$, and $p_0 \in \{-1, 0, -1\}$, we have $p_0 = 0$.
The division algorithm for integral polynomials says that $$P(x) = (x^p-x)Q(x) + R(x)$$ for some $Q(x) = \sum_{j=0}^{n-p}q_jx^j$, $R(x) = \sum_{j=0}^{p-1}r_jx^j \in \mathbb{Z}[x]$, where the degree of $R(x)$ is less than $p$.  Since both $P(x)$ and $x^p-x$ vanish modulo $p$, so does $R(x)$.  Because $R(x)$ has $p$ zeros modulo $p$ but its degree is less than $p$, by Lagrange's theorem, all coefficients of $R(x)$ must be divisible by $p$.
We now show that the coefficients of $R(x)$ have absolute values at most $p-1$, and hence must all be zeros, i.e., $R(x) = 0$.   

Expanding term by term both sides of  $P(x) = (x^p-x)Q(x) + R(x)$, we have
\begin{eqnarray*}
\sum_{j=0}^n p_jx^j &=& \sum_{j=0}^{n-p}q_jx^{j+p} -\sum_{j=0}^{n-p}q_jx^{j+1}  + \sum_{j=0}^{p-1}r_jx^j\\
&=&\sum_{j=n-p+2}^n q_{j-p}x^j + \sum_{j=p}^{n-p+1} (q_{j-p} - q_{j-1})x^j + \sum_{j=1}^{p-1}(r_j - q_{j-1})x^j + r_0.\\
\end{eqnarray*}
Comparing coefficients on the left and right sides, we see that the constant coefficient of $R(x)$ is 0, and the absolute values of the leftmost $p-1$ coefficients of $Q(x)$ are bounded by 1, because
\begin{eqnarray*}
r_0&=&p_0 = 0\\
|q_{j-p}|&=&|p_{j}| \le 1, \text{\ \ \ \ }\ \ n-p+2 \le j \le n\\
\end{eqnarray*}
Similarly, we see that the next leftmost $p-1$ coefficients of $Q(x)$ have absolute values bounded by 2, because their differences with the first leftmost $p-1$ coefficients are bounded in absolute value by 1, as can be seen below after changing the range of $j$:
\begin{eqnarray*}
|q_{j-p} - q_{j-1}|&\le& 1, \text{\ \ \ \ }\ \ n-2p+3 \le j \le n-p+1\\
|q_{j-2p+1} - q_{j-p}|&\le& 1, \text{\ \ \ \ }\ \ n-p+2 \le j \le n\\
\end{eqnarray*}

Repeating this argument, we conclude that the next leftmost $p-1$ coefficients of $Q(x)$ have absolute values bounded by 3, and so on.  Since there are $\lceil \frac{n}{p-1} \rceil$ such groups (the constant coefficient is zero), we conclude that the absolute values of coefficients of $R(x)$ are bounded by $\lceil \frac{n}{p-1} \rceil$ and hence by $p-1$ because
\begin{eqnarray*}
  \left\lceil \frac{n}{p-1}\right\rceil&\le& \left\lceil\frac{n}{1+\sqrt{n} - 1}\right\rceil\\
  &=& \left\lceil\sqrt{n}\right\rceil\\
  &<& 1 + \sqrt{n}\\
  &\le& p.
\end{eqnarray*}
\end{proof}
\begin{example}
Let $s = 0000\ 1010\ 0000\ 0000$ and $t = 0000\ 0000\ 0010\ 1000$.  The associated polynomials are $s(x) = x^{11}+x^9$ and $t(x) = x^5+x^3$.  The polynomial $P(x) = s(x) - t(x) = x^{11}+x^9 - x^5-x^3$ has the following irreducible factors over the rationals:
\begin{eqnarray*}
  P(x) = x^{11}+x^{9} - x^5-x^3&=&x^3(x-1)(x+1)(x^2+1)(x^2-x+1)(x^2+x+1).
\end{eqnarray*}
Primes 5, 7, 11 are all greater than $1 + \sqrt{11}$.  Since
\begin{eqnarray*}
  x^5-x = x(x-1)(x+1)(x^2+1) &|& P(x),\\
  x^7-x = x(x-1)(x+1)(x^2-x+1)(x^2+x+1) &|& P(x),\\
  x^{11}-x= x(x-1)(x+1)(x^4-x^3+x^2-x+1)(x^4+x^3+x^2+x+1) &\not|& P(x),
\end{eqnarray*}
it follows from Lemma \ref{lemma:3} that $P(x)$ vanishes modulo 5 and 7, but not modulo 11.  Exhaustive search shows that $d_H(s, t) = 9$ while $d_\forall(s, t) = 4$.
\end{example}

We are now ready to prove an $O(\sqrt{n}\log n\log\log n)$ upper bound on  $D_H(n)$ using Lemma \ref{lemma:3}.  Let $s$ and $t$ be distinct binary strings of length $n$, and let $P(x) = s(x)-t(x)$, whose degree is at most $n-1$. If $P(x)$ vanishes modulo $p$ for some $p \ge 1 + \sqrt{n-1}$, then $P(x)$ is divisible by $x^p-x$, which is in turn divisible by $\Phi_{p-1}(x)$.  Since these cyclotomic polynomials are co-prime, the sum $\delta$ of their degrees is bounded above by $n-1$, which implies the stated upper bound on $d_H(s, t)$.
\begin{theorem}
$D_H(n) \in O(\sqrt{n}\log n\log\log n)$.
\end{theorem}

\begin{proof}
Let $s$ and $t$ be distinct binary strings of length $n$, and suppose the Horner distance $d_H(s, t)$ is $M\sqrt{n}$, so that $P(x) = s(x) - t(x)$, whose degree is at most $n-1$, is congruent to the zero polynomial mod $p$ for all primes $p < M\sqrt{n}$. 

By Lemma \ref{lemma:2}, $M\sqrt{n} \le 2n$, so $M \le 2\sqrt{n}$.  By Lemma \ref{lemma:3}, for each such prime $p > \sqrt{n}$, $P(x)$ is divisible by $x^{p}-x = x(x^{p-1}-1)$ and hence divisible by the cyclotomic polynomial $\Phi_{p-1}(x)$, whose degree is $\phi(p-1) \in \Omega((p-1)/\log\log (p-1))$.  Since these cyclotomic polynomials are co-prime, the sum $\delta$ of their degrees is at most $n-1$.  There are approximately 
$$\frac{M\sqrt{n}}{\ln (M\sqrt{n})} - \frac{\sqrt{n}}{\ln\sqrt{n}} > \frac{M\sqrt{n}}{\ln (2\sqrt{n}\sqrt{n})} - \frac{\sqrt{n}}{\ln\sqrt{n}} \ge \alpha \frac{M\sqrt{n}}{\log n}$$ primes $p$ in the range $[\sqrt{n}, M\sqrt{n}]$ by the Prime Number Theorem, for some constant $\alpha$.  Because the function $x/\log\log x$ is eventually increasing, each $\Phi_{p-1}(x)$ contributes at least $\beta\sqrt{n}/\log\log(\sqrt{n})$, for some constant $\beta$,
 to the degree sum $\delta$, which is at most $n-1$, so $$\left(\frac{\alpha M\sqrt{n}}{\log{n}}\right)\left(\frac{\beta\sqrt{n}}{\log\log n}\right) = \left(\frac{\alpha\beta Mn}{\log n\log\log n}\right)\le \delta < n.$$ It follows that $M \in O(\log n\log\log n)$, and hence $d_H(s, t) \in O(\sqrt{n}\log n\log\log n)$.  Since this holds for arbitrary pairs of distinct binary strings of length $n$, we conclude that $D_H(n) \in O(\sqrt{n}\log n\log\log n)$.
\end{proof}

Because the more restrictive separation distance $D_H(n)$ is an upper bound on $D_{\forall}(n)$, we obtain our main result.
\begin{theorem}
$D_{\forall}(n) \in O(\sqrt{n}\log n\log\log n)$.
\end{theorem}

\section{Conclusion}
We show how to use Horner automata $H_{b, m}$ to $\forall$-separate two distinct binary strings of length $n$ and establish  almost matching lower and upper bounds on the minimum value of such $m$.  Closing the gap between these two bounds is an interesting open problem, as is the question of whether other families of automata can be designed to achieve better lower and upper bounds on the $\forall$-separation distance.

\section*{Acknowledgments}
Detailed comments and suggestions by the anonymous referees help improve the presentation of this paper.

\bibliographystyle{./eptcs}
\bibliography{ref}
\end{document}